\documentclass[runningheads]{llncs}

\usepackage{hyperref}
\usepackage{amsmath}
\usepackage{microtype} \usepackage{todonotes}

\newtheorem{openproblem}{Open Problem}

\graphicspath{{./figures/}}

\newbox{\myorcidauthbox}
\sbox{\myorcidauthbox}{\large\includegraphics[height=1.7ex]{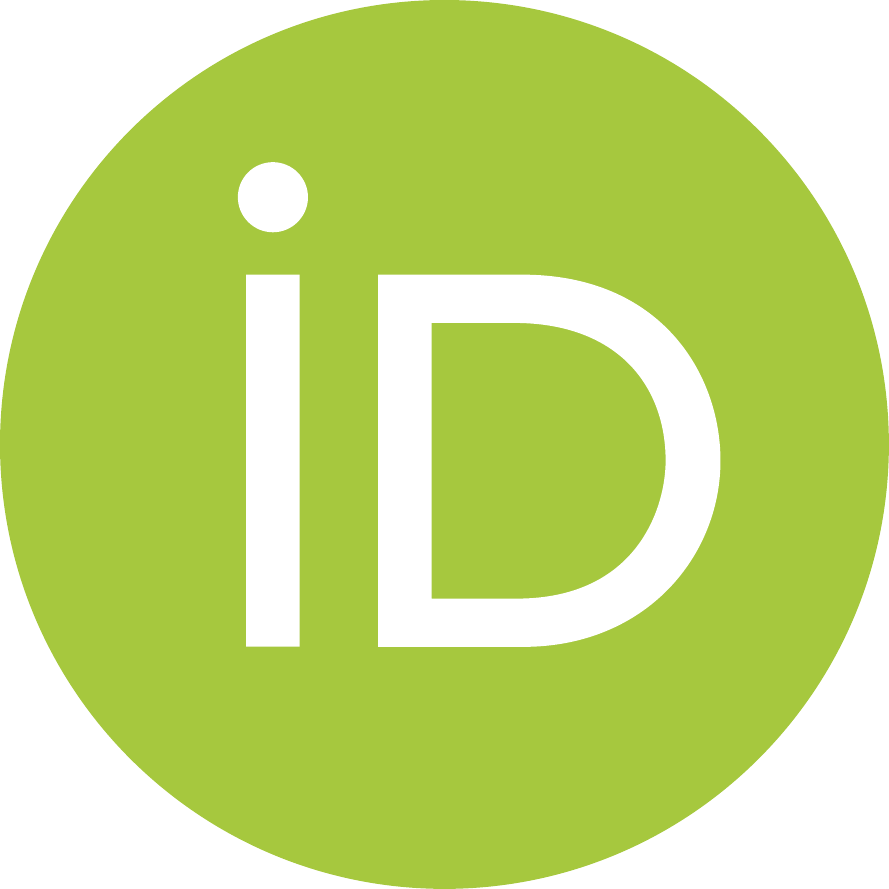}}
\newcommand{\orcid}[1]{\href{https://orcid.org/#1}{\usebox{\myorcidauthbox}}}

\renewcommand{\times}{x}
\usepackage[inline]{enumitem}

\title{Shooting Stars in Simple Drawings of~$K_{m,n}$\thanks{O.A., I.P., and A.W. were partially supported by the Austrian Science Fund (FWF) grant W1230.
A.G. was supported by MINECO project MTM2015-63791-R and Gobierno de Arag\'on under Grant E41-17 (FEDER).
I.P.\, and B.V. were partially supported by the Austrian Science Fund within the collaborative DACH project \emph{Arrangements and Drawings} as FWF project \mbox{I 3340-N35}. 
I.P. was supported by the Margarita Salas Fellowship funded by the Ministry of Universities of Spain and the European Union (NextGenerationEU).\\
This work initiated at the 6th Austrian-Japanese-Mexican-Spanish Workshop on Discrete Geometry which took place in June 2019 near Strobl, Austria. 
We thank all the participants for the great atmosphere and fruitful discussions.}
}
\author{Oswin~Aichholzer\inst{1}\orcid{0000-0002-2364-0583} \and
	Alfredo Garc\'ia\inst{2}\orcid{0000-0002-6519-1472}
	Irene~Parada\inst{3,4}\orcid{0000-0003-3147-0083} \and
	Birgit~Vogtenhuber\inst{1}\orcid{0000-0002-7166-4467} \and
	Alexandra~Weinberger\inst{1}\orcid{0000-0001-8553-6661}
}
\authorrunning{O. Aichholzer et al.}

\institute{Institute of Software Technology, Graz University of Technology, Graz, Austria\\
	\email{\{oaich,bvogt,weinberger\}@ist.tugraz.at} \and
	Departamento de M\'etodos Estad\'isticos and IUMA, Universidad de Zaragoza, Zaragoza, Spain\\
	\email{olaverri@unizar.es} \and
	Departament de Matem\`atiques, Universitat Polit\`ecnica de Catalunya, Barcelona, Spain\\
	\email{irene.maria.de.parada@upc.edu} \and
	Department of Information and Computing Sciences, Utrecht University, Utrecht, The Netherlands
}

\begin{document}

\maketitle

\begin{abstract}
	Simple drawings are drawings of graphs in which two edges have at most one common point (either a common endpoint, or a proper crossing). 
	It has been an open question whether every simple drawing of a complete bipartite graph $K_{m,n}$ contains a plane spanning tree as a subdrawing. 	
	We answer this question to the positive by showing that for every simple drawing of~$K_{m,n}$ and for every vertex $v$ in that drawing, 
	the drawing contains a \emph{shooting star rooted at~$v$}, that is, a plane spanning tree containing all edges incident to~$v$.  
\end{abstract}

\section{Introduction}
A \emph{simple drawing} is a drawing of a graph on the sphere $S^2$ or, equivalently, in the Euclidean plane where 
\begin{enumerate*}[label=(\arabic*)] 
\item the vertices are distinct points in the plane,
\item  the edges are non-self-intersecting continuous curves connecting their incident points, 
\item no edge passes through vertices other than its incident vertices, 
\item\label{prop:simple} and every pair of edges intersects at most once, either in a common endpoint, or in 
	the relative interior of both edges, forming a proper crossing.
\end{enumerate*} 
Simple drawings are also called \emph{good drawings}~\cite{AMRS18,EG_1973} or \emph{(simple) topological graphs}~\cite{kyncl2009,PST03}. 
In \emph{star-simple drawings}, the last requirement is softened so
that 
edges without common endpoints are allowed to cross several times. 
Note that in any simple or star-simple drawing, there are no tangencies between edges and incident edges do not cross.
If a drawing does not contain any crossing at all, it is called \emph{plane}.

The search for plane subdrawings of a given drawing has been a widely considered topic for simple drawings of the complete graph $K_{n}$ which still holds tantalizing open problems. 
For example, Rafla~\cite{rafla} conjectured that every simple drawing of $K_n$ contains a plane Hamiltonian cycle, a statement which is by now known to be true for $n \leq 9$~\cite{aafhpprsv-agdsc-15} and several classes of simple drawings (e.g., 2-page book drawings, monotone drawings, cylindrical drawings), but still remains open in general.
A related question concerns the least number of pairwise disjoint edges in any simple drawing of $K_n$. 
The currently best lower bound is $\Omega(n^{1/2})$\cite{twisted_socg}, which is improving over several previous bounds~\cite{bound_2009,bound_2014,Fulek:2013:TGE:2493132.2462394,PST03,Pach:2003:DET:2164311.2164326,RuizVargas17,Suk2013}, while the trivial upper bound of $n/2$ would be implied by a positive answer to Rafla's conjecture.
A structural result of Fulek and Ruiz-Vargas~\cite{Fulek:2013:TGE:2493132.2462394} implies that every simple drawing of~$K_n$ contains a plane sub-drawing with at least $2n-3$ edges.

We will focus on plane trees. Pach et al.~\cite{PST03} proved that every simple drawing of $K_n$ contains a plane drawing of any fixed tree with at most $c\log^{1/6}n$ vertices. For paths specifically, every simple drawing of $K_n$ contains a plane path of length $\Omega(\frac{\log n}{\log \log n})$~\cite{twisted_socg,unavoidable}. Further, it is trivial that simple drawings and star-simple drawings of $K_n$ contain a plane spanning tree, because every vertex is incident to all other vertices and adjacent edges do not cross. Thus, the vertices together with all edges incident to one vertex form a plane spanning tree. We call this subdrawing the \emph{star} of that vertex.

In this work, we consider the search for plane spanning trees in drawings of complete bipartite graphs. 
Finding plane spanning trees there is more involved than for~$K_n$.
In fact, not every star-simple drawing of a complete bipartite graph contains a plane spanning tree; see Fig.~\ref{fig:star}.

\begin{figure}[htb]
	\centering
	\includegraphics[page=1]{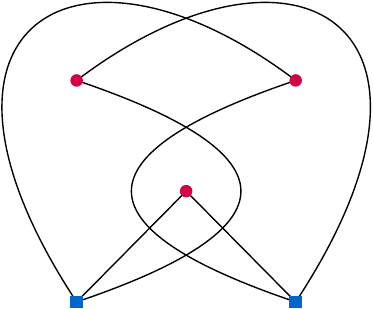}
	\caption{Star-simple drawing of $K_{2,3}$ that does not contain a plane spanning tree.}
	\label{fig:star}
\end{figure}

It is not hard to see that straight-line drawings of complete bipartite graphs
always contain plane spanning trees.  
Consider the star of an arbitrary vertex~$v$. 
The prolongation of these edges creates a set of rays originating at $v$ that partitions the plane into wedges, 
which we divide into two parts using the angle bisectors. 
We connect the vertices in each part of a wedge
to the point on the ray that bounds it.
These connections together
with the star of~$v$ form a plane spanning tree 
of a special type called \emph{shooting star}. 
A \emph{shooting star rooted at~$v$} is a plane spanning tree with root~$v$ that has height 2 and contains the star of vertex~$v$.
Aichholzer et al. showed in~\cite{euroCG_version_19} that simple drawings of $K_{2,n}$ and $K_{3,n}$, as well as so-called \emph{outer drawings} of $K_{m,n}$, always contain shooting stars. 
\emph{Outer drawings of $K_{m,n}$}~\cite{CardinalFelsner2018} are simple drawings in which all vertices of one 
bipartition class  
lie on the outer boundary.

\medskip \noindent{\bf Results.\, }
We
show in Section~\ref{sec:proof_shooting} that every simple drawing of $K_{m,n}$ contains shooting stars rooted at an arbitrary vertex of~$K_{m,n}$.
The tightness of the conditions is shown in Section \ref{sec:tightness} and 
in Section~\ref{sec:algo_remarks} we discuss algorithmic aspects.

\section{Existence of Shooting Stars}
\label{sec:proof_shooting}

In this section, we prove our main result, the existence of shooting stars: 

\begin{theorem}\label{theorem1}
	Let $D$ be a simple drawing of $K_{m,n}$ and 
	let $r$ be an arbitrary vertex of $K_{m,n}$. 
	Then $D$ contains a shooting star rooted at $r$.
\end{theorem}

\begin{proof}
We can assume that $D$ is drawn 
on a point set $P = R \cup B$, $R=\{r_1,\ldots ,r_{m}\}$, $B=\{ b_1,\ldots ,b_{n}\} $, 
in which the points in the two  
bipartition classes $R$ and $B$ 
are colored red and blue, respectively.
Without loss of generality let $r= r_1$. 

\begin{figure}[tb]
	\centering
	\includegraphics[page=2]{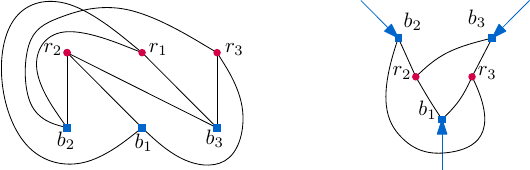}
	\caption{Left: Simple drawing of $K_{3,3}$. Right: its stereographic projection from $r_1$.} \label{fig:projection}
\end{figure}

To simplify the figures, we consider the drawing $D$ on the sphere
and apply a stereographic projection from $r$ onto a plane. 
In that way, the edges in the star of $r$ are represented as (not necessarily straight-line) infinite rays; see Fig.~\ref{fig:projection}. 
We will depict them in blue. In the following, we consider all edges oriented from their red to their blue endpoint.
To specify how two edges cross each other, we introduce some notation. 
Consider two crossing edges $e_1=r_ib_k$ and  $e_2=r_jb_l$ and let $\times$ be their crossing point. Consider the arcs $\times r_i$ and $\times b_k$ on $e_1$ 
and $\times r_j$ and $\times b_l$ on~$e_2$.  
We say that $e_2$ crosses $e_1$ in \emph{clockwise direction} 
if the clockwise cyclic order of these arcs around the crossing $\times$ is $\times r_i$, $\times r_j$, $\times b_k$, and $\times b_l$. Otherwise, we say that $e_2$ crosses $e_1$ in \emph{counterclockwise direction}; see Fig.~\ref{fig:clockwise}. \begin{figure}[t]
	\centering
	\includegraphics[page=1]{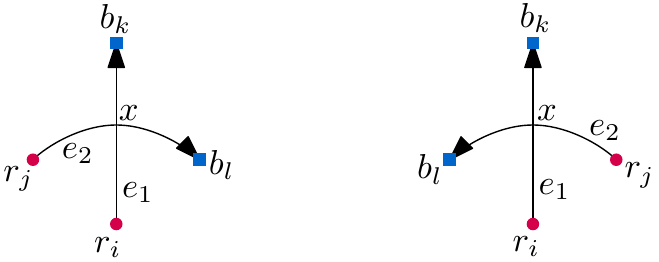}
	\caption{Left: $e_2$ crosses $e_1$ in clockwise direction. Right: $e_2$ crosses $e_1$ in counterclockwise direction.}
\label{fig:clockwise}
\end{figure}

We prove Theorem~\ref{theorem1} by induction on $n$. 
For $n=1$ and any $m\geq 1$, the whole drawing $D$ is a shooting star rooted at any vertex, and in particular at $r$. 

Assume that the existence of shooting stars rooted at any vertex has been proven for any simple drawing of $K_{m,n'}$ with $n'<n$.
By the induction hypothesis, the subdrawing of~$D$ obtained by deleting the blue vertex $b_1$ and its incident edges contains at least one shooting star rooted at~$r$.
Of all such shooting stars, let $S$ be one whose edges have the minimum number of crossings with $r{b_1}$, 
and let $M$ be the set of edges of $S$ that are not incident to $r$. 
We will show that $S\cup\{r{b_1}\}$ is plane and hence forms the desired shooting star.
Note that it suffices to show that $M\cup\{r{b_1}\}$ is plane, since $r{b_1}$ cannot cross any edges of $\{ \bigcup _{j=2}^n rb_{j}\}$ in any simple drawing.

Assume for a contradiction that $r{b_1}$ crosses at least one edge in $M$. 
When traversing $rb_1$ from $b_1$ to $r$, 
let $x$ be the first crossing point of $rb_1$ with an edge $r_kb_t$ in $M$. 
W.l.o.g., when orienting $rb_1$ from $r$ to $b_1$ and  
$r_kb_t$ from $r_k$ to $b_t$, $r_kb_t$ crosses $rb_1$ in counterclockwise direction
(otherwise we can mirror the drawing). 

Suppose first that the arc $r_kx$ (on $r_kb_t$ and oriented from $r_k$ to $x$) is crossed in counterclockwise direction 
by an edge incident to $b_1$ (and oriented from the red endpoint to $b_1$). 
Let $e=r_lb_1$ be such an edge whose crossing with $r_kx$
at a point $y$ is the closest to $x$. 
Otherwise, let $e$ be the edge $r_kb_1$ and $y$ be the point~$r_k$. 
In the remaining figures, we represent in blue the edges of the star of $r$, 
in red the edges in $M$, and in black the edge $e$.  

We distinguish two cases depending on whether $e$ crosses an edge of the star of $r$. 
The idea in both cases is to define a region $\Gamma$ and, inside it, 
redefine the connections between red and blue points to reach a contradiction.  

	\begin{figure}[thb]
		\centering
		\includegraphics[page=1]{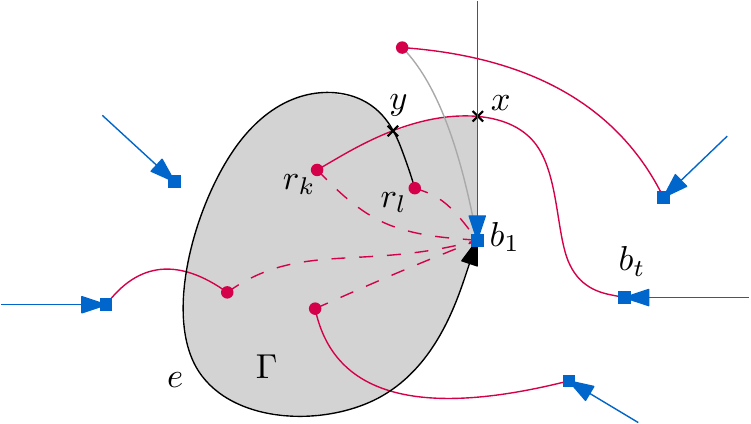}
		\caption[Illustration of Case 1  in the proof of Theorem~\ref{theorem1}.]{Illustration of Case 1.}
		\label{fig:Case1}
	\end{figure}
\smallskip
	\noindent{\bf Case 1: } 
	$e$ does not cross any edge of the star of $r$.
	Let $\Gamma$ be the closed region of the plane bounded by the arcs $yb_1$ (on $e$), $b_1x$ (on $rb_1$), and $xy$ (on $r_kb_t$); see Fig.~\ref{fig:Case1}. 
	Observe that all the blue points $b_j$ lie outside the region $\Gamma$ and that for all the red points $r_i$ inside region $\Gamma$, the edge $r_ib_1$ must be in $\Gamma$. 
	Let $M_\Gamma$ denote the set of edges $r_ib_1$ with $r_i\in \Gamma$ and 
	note that $r_kb_1\in M_\Gamma$. 
	Consider the set $M'$ of red edges obtained from $M$ by replacing, 
	for each red point $r_i\in \Gamma$, 
	the (unique) edge incident to $r_i$ in $M$ by the edge $r_ib_1$ in $M_\Gamma$, 
	and keeping the other edges in $M$ unchanged. 
	In particular, the edge $r_kb_t$ has been replaced by the edge $r_kb_1$. 
	The edges in $M_\Gamma$ neither cross each other nor cross any of the blue edges $r{b_j}$. 
	Moreover, we now show that the non-replaced edges in $M$ must lie completely outside $\Gamma$.  
	These edges can neither cross $r_kb_t$ (by definition of $M$) nor the arc $b_1x$ (on $rb_1$).  
	Thus, if they are incident to $b_1$, they cannot cross the boundary of~$\Gamma$;
	If they are not incident to $b_1$, both their endpoints lie outside $\Gamma$ and they can only cross the boundary of $\Gamma$ at most once (namely, on the arc $b_1y$).  
Therefore, $M'$ satisfies that $M'\cup \{ \bigcup _{j=2}^n r{b_{ j}}\}$ is plane and has fewer crossings with~$r{b_1}$ than $M$, since at least the crossing $x$ has been eliminated and no new crossings have been added.
This contradicts the definition of $M$ as the one with the minimum number of crossings with $rb_1$.

\smallskip	
\noindent{\bf Case 2: } 
	$e$ crosses the star of $r$. 
When traversing $e$ from $r_k$ or $r_l$ (depending on the definition of $e$) to $b_1$, 
	let $I=\{\alpha,\beta,\ldots ,\rho\} $ be the indices of the edges of the star of $r$ in the order as they are 
	crossed by~$e$ and let $y_{\alpha},\ldots ,y_{\rho}$ be the corresponding crossing points on $e$. 
	Note that, when orienting $e$ from $r_k$ or $r_l$ to $b_1$, the edges $rb_\xi, \xi \in I$, 
	oriented from $r$ to $b_\xi$, cross $e$ in counterclockwise direction, 
	since they can neither cross $r_kb_t$ (by definition of $M$) nor $rb_1$.

	 The three arcs $ry_{\alpha}$ (on $rb_\alpha$), $y_{\alpha}b_1$ (on $e$), and $b_1r$ 
	 divide the plane into two (closed) regions, 
	 $\Pi_\text{left}$, containing vertex $r_k$, and $\Pi_\text{right}$, containing vertex $b_t$. 
	 For each $\xi\in I$, let $M_\xi$ be the set of red edges of $M$ incident to some red point in $\Pi_\text{right}$ and to $b_{\xi}$. 
	 Note that all the edges in $M_\xi$ (if any) must cross the edge~$e$. 
	 When traversing $e$ from $r_k$ or $r_l$ to $b_1$,
	 we denote by $x_\xi,z_\xi$ the first and the last crossing points of $e$ with the edges of $M_\xi\cup r{b_\xi}$, 
	 respectively; see Fig.~\ref{fig:Case2} for an illustration.   	
	 We remark that both $x_\xi$ and $z_\xi$ might coincide with $y_\xi$ and, in particular, if $M_\xi=\emptyset $ then $x_\xi=y_\xi=z_\xi$.\begin{figure}[htb]
	 	\centering
	 	\includegraphics[page = 1]{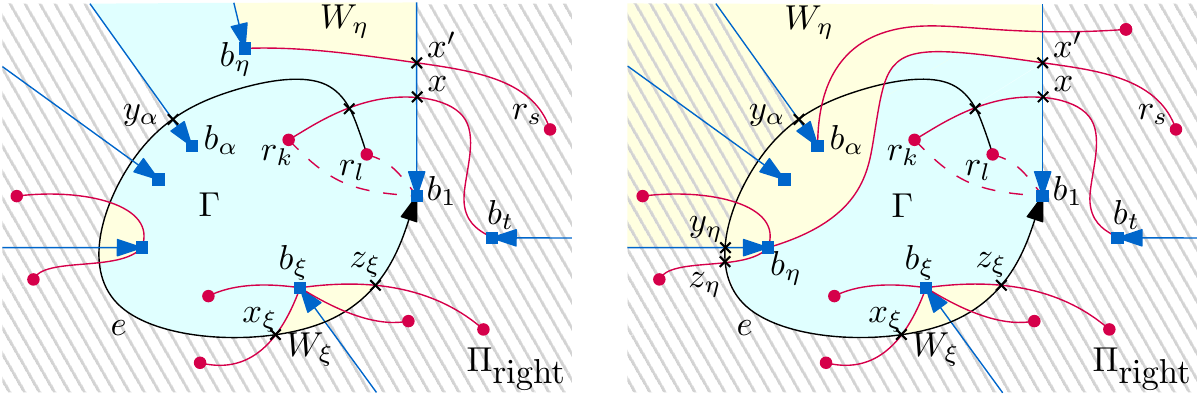}
	 	\caption[Illustration of Case 2 in the proof of Theorem~\ref{theorem1}.]
	 	{Illustration of Case 2. 
	 		Region $\Pi_\text{right}$ is striped in gray, region $\Gamma$ is shaded in blue, and regions in $\bigcup_{\xi\in I} W_\xi \cup W_{\eta}$ are shaded in yellow. 
	 		Left: $b_\eta$ does not cross $e$ ($\eta\notin I$). Right: $b_\eta$ crosses $e$ ($\eta\in I$).
}
	 	\label{fig:Case2} 
	 \end{figure}
	 
	 We now define some regions in the drawing $D$. 
	 Suppose first that there are edges in~$M$ (oriented from the red to the blue point) 
	 that cross $r{b_1}$ (oriented from $r$ to $b_1$) in clockwise direction. 
	 Let $r_sb_{\eta}$ be the edge in $M$ whose clockwise crossing with $r{b_1}$
	 at a point $x'$ is the closest one to $x$ (recall that the arc $b_1x$ on $rb_1$ is not crossed by edges in $M$). 
	 Then, if $\eta \notin I$, we denote by $W_{\eta}$ the region bounded by the arcs 
	 $rx'$ (on $rb_1$), $x'b_{\eta}$ (on $r_sb_\eta$), and $rb_{\eta}$ and not containing $b_1$; 
	 see Fig.~\ref{fig:Case2}~(left). If $\eta\in I$, we define $W_{\eta}$ as the region bounded by the arcs  
	 $rx'$ (on $rb_1$), $x'b_{\eta}$ (on $r_sb_\eta$), $b_{\eta}z_{\eta}$, $z_{\eta}y_{\eta}$ (on~$e$), and $y_{\eta}r$ (on~$rb_\eta$) and not containing~$b_1$;  
	 see Fig.~\ref{fig:Case2}~(right). If no edges in $M$ cross $r{b_1}$ in clockwise direction, then $\eta$ is undefined and
	 we set $W_{\eta}= \emptyset$ for convenience. Moreover, for each $\xi \in I\setminus \{\eta\}$, 
	 we define $W_\xi$ as the region bounded by the arcs $x_\xi b_\xi$, $b_\xi z_\xi$, and $z_\xi x_\xi$ 
	 (and not containing $b_1$); see again Fig.~\ref{fig:Case2}. 

	 We can finally define the region $\Gamma$ for Case 2, which is the region obtained from $\Pi_\text{left}$ by removing the interior of all the regions $W_\xi$, $\xi\in I$ 
	 plus region $W_{\eta}$ if $\eta\notin I$ (otherwise it is already contained in $\bigcup_{\xi\in I} W_\xi$). Now consider the set of red and blue vertices contained in 
$\Gamma$. 
	 Let $J$ denote the set of indices such that for all $j\in J$, the blue point $b_j$ lies in $\Gamma$ 
	 (note that $1\in J$). 
	 Since $b_t$ is not in $\Gamma$, we can apply the induction hypothesis to the subdrawing of $D$ induced by the vertices in $\Gamma$ plus $r$.
Hence there exists a
set of edges $M_\Gamma$ connecting 
	 each red point in $\Gamma$ with a blue point $b_{j}$, $j\in J$ such that
	 $M_\Gamma\cup \{ \bigcup _{j\in J}r{b_{j}}\} $ is plane. Moreover, all the edges in $M_\Gamma$ lie entirely in~$\Gamma$:
	 An edge in $M_\Gamma$ cannot cross any of the edges $rb_j$, with $j\in J$. 
	 Thus, it cannot leave $\Pi_\text{left}$, as otherwise it would cross $e$ twice. 
	 Further, if it entered one of the regions in $\bigcup_{\xi\in I} W_\xi \cup W_{\eta}$, it would have to leave it crossing $e$, 
	 and then it could not re-enter $\Gamma$.

	Consider the set $M'$ of red edges obtained from $M$ by replacing, 
	for each red point $r_i\in \Gamma$, 
	the edge $r_ib_\xi$ in $M$ by the edge $r_ib_j$, $j\in J$, in $M_\Gamma$, 
	and keeping the other edges in $M$ unchanged. 
	In particular, the edge $r_kb_t$ has been replaced by some edge $r_kb_j$, $j\in J$. 
	The edges in $M_\Gamma$ neither cross each other nor cross any of the blue edges $r{b_j}$, $j\in J$ 
	nor any of the other ones, lying completely outside $\Gamma$. 	
	Moreover, the non-replaced edges in $M$ cannot enter $\Gamma$ 
	since the only boundary part of $\Gamma$ that they can cross are arcs on $e$.  
	Therefore, $M'$ satisfies that $M'\cup \{ \bigcup _{j=2}^n r{b_{ j}}\}$ is plane and has fewer crossings with $r{b_1}$ than $M$.
This contradicts the definition of $M$ as the one with the minimum number of crossings with $rb_1$.
\qed
\end{proof}

\section{Some Observations on Tightness}
\label{sec:tightness}
There exist simple drawings of $K_{m,n}$ in which every plane subdrawing has at most as many edges as a shooting star. For example, consider a straight-line drawing of $K_{m,n}$ where all vertices are in convex position such that all red points are next to each other in the convex hull; see Fig.~\ref{fig:tight}~(left). The convex hull is an $(m+n)$-gon which shares only two edges with the drawing of $K_{m,n}$; see Fig.~\ref{fig:tight}~(right). All other edges of the drawing of $K_{m,n}$ are diagonals of the polygon. As there can be at most $(m+n)-3$ pairwise non-crossing diagonals in a convex $(m+n)$-gon, any plane subdrawing of this drawing of $K_{m,n}$ contains at most $m+n-1$ edges.

\begin{figure}[t]
		\centering
		\includegraphics[]{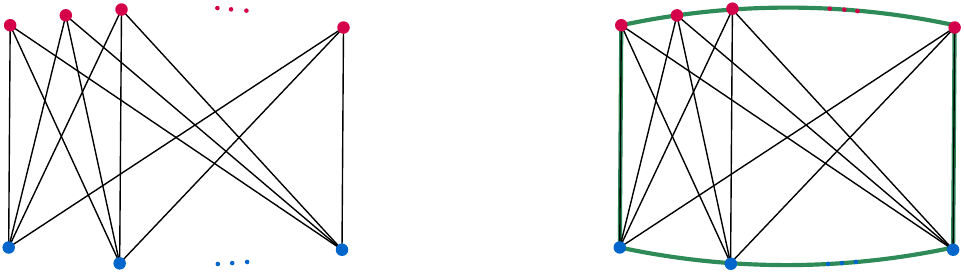}
		\caption{Left: A simple drawing of $K_{m,n}$ where no plane subdrawing has more edges than a shooting star. Right: convex $(n+m)$-gon on the convex hull (green).
}
		\label{fig:tight}
	\end{figure}

Furthermore, both requirements from Theorem~\ref{theorem1}---the drawing being simple and containing a complete bipartite graph---are in fact necessary:
As mentioned in the introduction, not all star-simple drawings of $K_{m,n}$ contain a plane spanning tree. 
Further, 
if in the example in Fig.~\ref{fig:tight} (left), 
we delete one of the two edges of $K_{m,n}$ on the boundary of the convex hull, 
then any plane subdrawing has at most $m+n-2$ edges and hence it cannot contain any plane spanning tree.

\section{Computing Shooting Stars}
\label{sec:algo_remarks}

The proof of Theorem~\ref{theorem1} contains an algorithm with which we can find shooting stars in given simple drawings. 
We start with constructing the shooting star for a subdrawing that is a~$K_{m,1}$ and then inductively add more vertices. Every time we are adding a new vertex, the shooting star of the step before is a set fulfilling all requirements of~$M_1\cup \{ \bigcup _{j=2}^n rb_{j}\} $ in the proof. By replacing edges as described in the proof, we obtain a new set with the same properties and fewer crossings. We continue replacing edges until we obtain a set of edges ($M$ in the proof) that form a shooting star for the extended vertex set. 
We remark that the runtime of this algorithm might be exponential, as finding the edges of~$M_\Gamma$ might require solving the problem for the subgraph induced by~$\Gamma$.
However, we believe that there exists a polynomial-time algorithm for this task.

\begin{openproblem}
	Given a simple drawing of $K_{m,n}$, is there a polynomial-time algorithm to find a plane spanning tree contained in the drawing?
\end{openproblem}

For some relevant classes of simple drawings of $K_{m,n}$ we can efficiently compute shooting stars. 
This is the case of outer drawings. 
In~\cite{euroCG_version_19} it was shown that these drawings contain shooting stars and this existential proof leads directly to a polynomial-time algorithm to find shooting stars in outer drawings. 
In Appendix~\ref{ap:monotone} we show that \emph{monotone drawings} of $K_{m,n}$, which are simple drawings in which all edges are $x$-monotone curves, admit an efficient algorithm for computing a shooting star. 
Fig.~\ref{fig:monotone} shows an illustration. 
The idea is as follows.
Let the sides of the bipartition be $R$ and $B$ and 
let $v$ be the leftmost vertex (without loss of generality assume $r\in R$).
We first consider the star of $r$, which we denote by $T$. 
For each vertex $w\in R$ not in $T$ we shoot two vertical rays, one up and one down. 
If only one of those vertical rays intersects $T$ we connect $w$ with the endpoint in $B$ of the first intercepted edge. 
If both vertical rays intersect $T$ we consider the endpoints in $B$ of the first edge intercepted by the upwards and the downwards ray. 
We connect $w$ with the horizontally closest one of the two.
If neither of the rays intersects $T$ we connect~$w$ with the horizontally closest vertex in $B$. 
In Appendix~\ref{ap:monotone} we prove that this indeed constructs a shooting star and we show how to efficiently compute it. 

\begin{figure}[htbp]\centering
	\includegraphics[page=2]{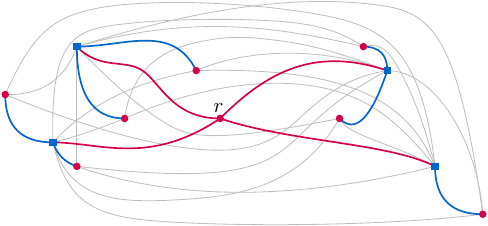}
	\caption{An example of a shooting star rooted at $r$ in a monotone drawing of~$K_{9,5}$.}
	\label{fig:monotone}
\end{figure}

\bibliography{shooting_stars.bib}

\newpage
\appendix

\section{Monotone Drawings}
\label{ap:monotone}

In this section we consider monotone drawings of $K_{m,n}$. 
We assume the information about the drawing is given as the rotation system (clockwise cyclic order of the edges around each vertex) together with the crossings sorted along each edge and the vertices sorted by $x$-coordinate. 
The \emph{linear separator} of a vertex $v$ in an $x$-monotone drawing is the vertical line going through the vertex. It separates the edges incident to $v$ into the set of edges going to vertices left of $v$ and the set of edges going to the vertices right of $v$. 
Note that the linear separators are implicitly given.

\begin{theorem}\label{thm:monotone}
Given a monotone drawing of the complete bipartite graph $K_{m,n}$ we 
	can compute a shooting star in linear time in the size of the input.
\end{theorem}

\begin{proof}
	Let the sides of the bipartition be $R$ (red vertices) and $B$ (blue vertices), and without loss of generality assume the leftmost vertex is $r\in R$. 
	We denote the star of $r$ by $T$. 
	We construct a set $M$ of edges such that $T \cup M$ forms a shooting star rooted at $r$. 
	For each red vertex $w\in R$ not in $T$, we shoot two vertical rays, one up ($r\!\uparrow$) and one down ($r\!\downarrow$). 
	There are three possible cases (see Fig.~\ref{fig:monotoneproof} for an illustration): 
\begin{figure}[htbp]\centering
		\includegraphics[page=2]{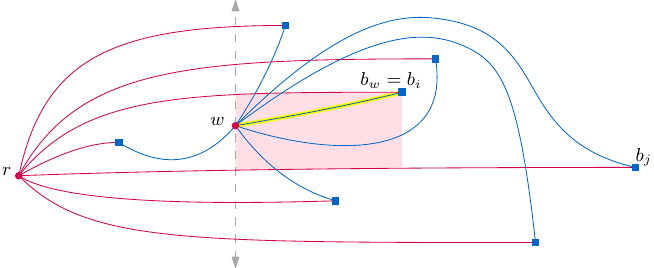}
		\caption{Construction of the shooting star in a monotone drawing. The star of $r$ ($T$) is shown in red. The edge incident to $w$ picked by the algorithm is highlighted. }
		\label{fig:monotoneproof}
	\end{figure}
\begin{itemize}
		\item[(i)] If only one of those vertical rays intersects $T$, 
		let $rb_i$ be the edge in $T$ producing the intersection that is closest to $w$, 
		that is, $b_i$ is the endpoint in~$B$ of the first intercepted edge in $T$. 
		We define $b_w := b_i$ and add $wb_w$ to $M$.
		\item[(ii)] If both vertical rays intersect $T$ we consider the endpoints $b_i$ and $b_j$ in $B$ of the first edge intercepted by $r\!\uparrow$ and by $r\!\downarrow$, respectively. 
		More precisely, $rb_i$ ($rb_j$) is the edge in $T$ producing the intersection with the upwards (resp. downwards) ray that is closest to $w$. 
		Let $b_w \in \{b_i, b_j\}$ be the point that is horizontally closest to $w$. 
		We add $wb_w$ to $M$.
		\item[(iii)] If neither of the rays intersects $T$, 
let $b_i$ be the horizontally closest vertex in~$B$. 
We define $b_w := b_i$ and add $wb_w$ to $M$. 
\end{itemize}
	
	We next prove that $T \cup M$ is indeed a shooting star and we show how to efficiently compute it. 
	The first part relies on the following claim:
	\begin{claim}
		Let $w\in R$ be a red vertex not in $T$. 
		The edge $wb_w$ does not cross $T$. 
	\end{claim}	
	\begin{proof}
		The proof for case (i) follows from the monotonicity property and the fact that the uncrossed ray, the part of the crossed ray until the first intersection, and the edge $rb_w$ cannot be crossed by any edge in $T$. 
		For case (ii), note that $wb_w$ must be contained in the region bounded by $rb_i$, $rb_j$, and the vertical lines through $w$ and $b_w$. 
		The boundary of this region cannot be crossed by any edge in $T$; see the shaded region in Fig.~\ref{fig:monotoneproof} for an illustration. 
		For case (iii), 
		since $r$ is the leftmost vertex, 
		$b_i$ is horizontally closer to $w$ than $r$.  
		Thus, 
		the horizontally closest point to $w$ in $T$ is $b_w = b_i$ and the statement follows immediately from monotonicity. 
		\qed
\end{proof}

	To prove that $T \cup M$ is indeed a shooting star it remains to show that no two edges in $M$ cross each other. 
	Consider a red vertex $w\in R$ in $T$ and the edge $wb_w$. 
By construction, the first intersection of a vertical ray up (or vertical ray down) from any point in the edge $wb_w$ with edges of $T$ is with the same edge if any. 
	This means that shooting a ray from any point in the edge $wb_w$ we find the same situation as the one that defines the case:  
	If $w$ falls under case (i), 
	the vertically closest edge in $T$ above or below $wb_w$ is, at any point, $rb_i$ and no other edge from $T$ lies on the opposite side;
	if $w$ falls under case (ii), 
	the vertically closest edges in $T$ above and below $wb_w$ are, at any point, $rb_i$ and $rb_j$, respectively; 
	and if $w$ falls under case (iii), 
	no edge from $T$ is above or below $wb_w$ at any point.
This implies that no two edges in $M$ can cross, since by the definition of simple drawings incident edges do not cross.

	For the algorithmic part, note that if $w$ fall under case (iii) this is easy to detect and $wb_w$ is easy to compute just using the horizontal sorting of the points (and the existence of crossings with $T$). 
	Otherwise, we consider the edges incident to $w$ on the right side of the linear separator sorted clockwise around $w$ (starting the sweeping from a vertical up direction). 
	The edge $wb_w$ is either the first or the last such edge that does not cross any edge from $T$. 
	More precisely, among those it is the one with the leftmost blue endpoint. 
	This allows to efficiently compute $M$ and therefore the shooting star.
	\qed
\end{proof}

\end{document}